\documentclass[11pt]{article}

\usepackage{amssymb}
\usepackage{amsmath}
\usepackage{amsthm}
\usepackage{amsfonts}
\usepackage{amsbsy}

\usepackage{enumerate}
\usepackage{hyperref}
 
\usepackage{graphicx}
\usepackage{comment}
\usepackage{color}
\usepackage{comment}
\includecomment{versiona}
\excludecomment{versionb}

\label{defis_LtX}

\newtheorem{theorem}{Theorem}[section]

\newtheorem{definition}{Definition}[section]
\newtheorem{remark}{Remark}[section]

\newcommand{\be}{\begin{equation}}
\newcommand{\ee}{\end{equation}}
\newcommand{\ba}{\begin{align}}
\newcommand{\ea}{\end{align}}
\newcommand{\PP}{\mathcal{P}}

\newcommand{\cl}{\mathcal{L}}
\newcommand{\La}{\Lambda}

\newcommand{\N}{\mathbb{N}}

\newcommand{\R}{\mathbb{R}}

\newcommand{\dd}{,\dots,}

\newcommand{\kk}{\mathbf{k}}
\newcommand{\e}{\mathbf{e}}

\newcommand{\Y}{\mathbf{Y}}
\newcommand{\ZZ}{\mathbb{Z}}
\newcommand{\NN}{\mathbb{N}}

\newcommand{\s}{\mathfrak{S}}

\newcommand{\lla}{\mathcal{L}_{\Lambda}}

\newcommand{\lv}{\left\vert}
\newcommand{\rv}{\right\vert}

\newcommand{\pa}{\partial}
\newcommand {\sSigma}{\mathfrak{S}}
\newcommand{\tg}{\tilde{g}}
\newcommand{\tA}{\tilde{A}}

\newcommand{\ry}{\mathrm{y}}

\newcommand{\rw}{\mathrm{w}}

\begin{document}

\title{Markov semigroups with hypocoercive-type generator in Infinite Dimensions II: Applications.}
\author{ V. Kontis \footnote{v.kontis@imperial.ac.uk,  
 Department of Epidemiology and Biostatistics
School of Public Health,
Imperial College London }, M. Ottobre \footnote{michelaottobre@gmail.com, Mathematics Institute, 
    Warwick University,  CV4 7AL, UK.} and B. Zegarlinski 
    \footnote{b.zegarlinski@imperial.ac.uk, Imperial College London, South Kensington Campus, London SW7 2AZ. Supported by Royal Society Wolfson RMA.} }

\date{\today}
\maketitle

 
\begin{abstract}
In this paper we show several applications of the general theory developed in \cite{MV_I}, where we studied smoothing and ergodicity for  infinite dimensional Markovian systems with hypocoercive type generator. 
   \end{abstract}

\section{Introduction} \label{Part II}


\label{S.4: Application of results} 
In \cite{MV_I} we studied infinite dimensional models of 
interacting dissipative systems   with hypocoercive-type generator
 and provided a basis of a general theory for controlling small and large time smoothing properties, existence of invariant measures and strong ergodicity. In particular we focussed on the following framework: Let $L$ be the Markov generator of a dissipative dynamics on $\R^m$, of the form
$$
L=Z_0^2+B,
$$
where $Z_0$ and $B$ are first order differential operators on $\R^m$, and suppose new fields are generated only through the interaction between the first and the second order part of the generator, i.e. for some $N\in \N$ there exist differential operators  $Z_1 \dd Z_N$ such that
$$
Z_{j+1}=[B, Z_j] \qquad \mbox{for all }j=0 \dd N-1.
$$ 
In \cite{MV_I}  we first studied the short and long time behaviour of the $n$-th order derivative of the semigroup $P_t\equiv e^{tL}$ generated by $L$, obtaining pointwise estimates for a suitable time dependent quadratic forms. We then considered infinitely many isomorphic copies of the generator $L$, each of them "placed" at a point of the lattice $\ZZ^d$, and let them interact, i.e. we looked at the dynamics 
on $(\R^m)^{\ZZ^d}$ generated by the Markov operator
$$
\cl=\sum_{x\in \ZZ^d}L_{x}+\sum_{y\in\ZZ^d}\sum_{i=0}^N q_{i,x}Z_{i,x}
+\sum_{y,y'\in\ZZ^d}\sum_{i,i' \in J}\s_{ii',yy'}Z_{i,y}Z_{i',y'},
$$ 
where $D_x$ denotes the isomorphic copy of a given a differential operator $D$  acting at the point $x\in \ZZ^d$ (for a more precise definition see \cite[Section 3]{MV_I}), $J$ is some subset of $I \equiv \{0 \dd N\}$ and $q_{i,y}= q_{i,y}(\omega), \s_{ii',yy'}= \s_{ii',yy'}(\omega), \omega\in  (\R^m)^{\ZZ^d}$ are interaction coefficients.  
 
After proving the well posedness of the semigroup generated by $\cl$, we analysed the smoothing properties of the infinite dimensional dynamics 
$\PP_t\equiv e^{t \cl}$ on $(\R^m)^{\ZZ^d}$, by estimating appropriate time dependent quadratic forms, inspired by the finite dimensional 
setting.  Finally, we studied the ergodicity of $\PP_t$, for which the equilibrium measure $\mu$ is not a  priori known. Existence of the invariant measure  is proved by Lyapunov function techniques, \cite[Section 4]{MV_I},
and uniqueness is obtained by methods analogous to those developed in \cite{DragoniKontisZegarlinski2011}. In particular we show that, once a Lyapunov function $\rho$ is known for the finite dimensional model $L$, $\tilde{\rho}=\sum_{x\in\ZZ^d}\rho_x$ is a Lyapunov function for $\cl$ (under some technical assumptions on the interaction functions, see \cite[Section 4]{MV_I}) and in \cite[Section 5]{MV_I} we provide several  criteria for the uniqueness of the equilibrium measure. 
 
 In the present paper we give a wealth of examples that belong to the framework described above. The organization and content of the paper is as follows. In Section \ref{SubS:Langevin equation} we describe an application of the theory developed in  \cite{MV_I} involving infinitely many interacting copies of Langevin-type dynamics. Here, we consider more general commutation relations than the ones described before and the finite dimensional generator arises in non-equilibrium Statistical Mechanics in the context of the  heat baths formalism, so that the infinite dimensional dynamics considered in Section \ref{SubS:Langevin equation}  can be interpreted as resulting from the interaction of infinitely many heat baths as well as interaction between other degrees of freedom. For this example we show in detail how to prove short time smoothing properties, exponential decay to equilibrium as well as existence  and uniqueness of the invariant measure. 
 In Section \ref{Sec3}, instead, we present the most straightforward
  application, i.e. Lie groups of Heisenberg type: in Section   
  \ref{ExpDecay4FilliformAlgebra_FullDilation} and Section 
  \ref{Filliform Algebras Partial Dilation} the finite dimensional model is given by filiform algebras with full and partial dilation, respectively;  Section  \ref{H_model_PartialDilation} describes the Heisenberg model 
with partial dilation, and we conclude in Section \ref{B-S Model} with the B-S model. For all these dynamics we explain how to prove exponential convergence to equilibrium under appropriate assumptions on the interaction functions, which we will detail in each case. 
   
        
\section{Langevin dynamics}\label{SubS:Langevin equation}
An example to which the theory presented so far may be applied comes from non-equilibrium statistical mechanics, in particular from the Generalized Langevin equation (GLE), which is  a popular model of  for a particle coupled to a heat bath (see \cite{Kup03}), 
\begin{equation}\label{GLE}
\ddot{q}(t)=-\partial_qV(q)-\int_0^t ds\,\gamma(t-s)\dot{q}(s) +F(t).
\end{equation}
In (\ref{GLE}) $q(t)$ represents the position of the distinguished particle (here $q(t)\in\R$ just  for simplicity, the equation can be rewritten in $\R^n$), $V=V(q)$ is a potential, $\gamma(t)$ is a smooth kernel and  $F(t)$ is a mean zero stationary Gaussian process. Noise and memory kernel are related through the following  \textit{fluctuation dissipation principle}
\begin{equation}\label{fluct dissp principle1}
E(F(t)F(s))\sim \gamma(t-s).
\end{equation}
The GLE can be derived by describing the system "particle + bath"
as a mechanical system in which a distinguished particle interacts with $r$ heat bath molecules through linear springs with random stiffness parameter, and then taking the thermodynamic limit $r\rightarrow \infty$. 
The resulting equation (\ref{GLE})  is in general non-Markovian, though for some specific choices of the correlation function $\gamma(t)$ it is {\em equivalent} to a system of SDEs in an extended state space. The most general system that can be obtained via this procedure is as follows (see \cite{Kup03}):
\begin{subequations}\label{MarkovianApprox}
\begin{eqnarray}
dq&=&p\, dt \\
dp&=&-\partial_qV(q)\,dt+g\cdot u \,dt    \\
du&=& (-p\,g-\mathcal{A}u) \,dt+C\,dW(t) ,  
\end{eqnarray}
\end{subequations}
where $(q,p)\in\R^{2}$, $u$ and $g$ are column vectors of $\R^{d}$, $\cdot $ denotes Euclidean scalar product,  $W(t)=(W_1(t), \dots, W_{d}(t))$ is a $d$-dimensional Brownian motion, $V(q)$ is a potential and $\mathcal{A}$ and $C$ are constant coefficients $d\times d$ matrices, related through the fluctuation dissipation principle, which in the present case reads
\be\label{fluct-dissiy}   
\mathcal{A}+\mathcal{A}^T=CC^T.
\ee
Notice that (\ref{MarkovianApprox})  is a degenerate O-U process, degenerate in the sense that the diffusion matrix is singular.   
Also, $\mathcal{M}:=CC^T$ is a semipositive definite symmetric matrix and we denote $m=\mbox{Rank}\,\mathcal{M}$.
For the vector $g$, we shall always assume that $g\neq 0$ (to avoid the uninteresting case in which there is no coupling between the heat bath and the particle). In the remainder of this section we shall denote $ v=(q,p,u)\in\R^N$,  $N=d+2$. If we assume, without loss of generality (see \cite{Kup03, mythesis}) that we are working in the coordinate system in which $\mathcal{M}$ is diagonal \footnote{Observe that in this coordinate system $a_{ii}\geq 0$.}, the generator of (\ref{MarkovianApprox}) is    
\be \label{generator}
L=p\partial_q-\pa_qV\pa_p+\sum_{i=1}^d g_i \left(u_i\pa_p - p  
\pa_{u_i}\right) -
\sum_{i,j=1}^d a_{ij}u_j\pa_{u_i}+\frac{1}{2}\sum_{j=1}^d \mathcal{M}_{jj}\pa^2_{u_j}.
\ee
When $V(q)$ is quadratic, the operator $L$ is precisely of the form 
\[L\equiv \sum_iZ_{0,i}^2 +B\]
once we set $$B=p\partial_q-\pa_qV\pa_p+\sum_{i=1}^d g_iu_i\pa_p-p\sum_{i=1}^dg_i\pa_{u_i}-
\sum_{i,j=1}^d a_{ij}u_j\pa_{u_i}.$$
Taking copies of such a generator at any point of $\ZZ^d$ and then adding an interaction term can be interpreted as considering infinitely many interacting systems of the type "particle+heat bath". 

It was shown in \cite{mythesis} that for the generator $L$ it is possible to find  a confining  function $\bar\rho(v)$ with compact level sets such that
$$
\bar{L}\bar\rho(v) \leq - a\bar\rho(v)+ d,
$$
for some $a,d>0$, i.e. the function  required for compactness condition sufficient for existence of invariant measure. 
Constructing such a function for general matrices $\mathcal{A}$ and $C$ requires introducing quite some notation, so here we will give an explicit expression for $\bar\rho$ only for the simplest case in which $d=1$. For the general case we refer the reader to \cite{mythesis}.

If $d=1$, assuming that $V(q)=q^2/2$, we have
\be \label{genera.1}
L= \pa^2_u +\left(p\partial_q- q\pa_p+  g \left(u \pa_p - p  \pa_u \right) \right)
-
\lambda u\pa_u  \equiv Z_0^2 +B -\lambda D_0
\ee
with
\be \begin{split}
&[B,Z_0] = -g\pa_p \equiv Z_1  \\
&[B,Z_1] =  g(\pa_q -g \pa_u )\equiv Z_2 \\
&[B,Z_2] = -(1+g^2) Z_1   \\
&[Z_i,Z_j] = 0\,.
\end{split}\ee
For this generator  the corresponding function $\bar\rho$ can be chosen as follows
\begin{eqnarray*}
\bar\rho =\bar C\,\left( q^2 + p^2 + u^2\right) +\bar R\, pq+
g\bar H\, pu
\end{eqnarray*}
with
$$
\bar R\ll \bar H \ll \bar C 
$$
so that for some $\bar c\in(0,\infty)$
\[\bar\rho \geq \bar c\,\left( q^2 + p^2 + u^2\right), \,\,\, \mbox{and} \,\,\,  \bar L\bar\rho \leq - a_0\bar\rho +b_0\]
with some $\bar a_{0} ,\bar b_0 \in(0,\infty)$.
\begin{theorem}\label{thm:exsinvmeas}
 Consider an infinite system with  generator
\be \label{8_generator} 
\mathcal{L} \equiv \sum_{x\in\ZZ^d}  L_x + \sum_{y\neq x\in\ZZ^d} G_{xy} u_y\pa_{u_x} +\sum_{i=0,1,2;x\in\ZZ^d} q_{i,x}  Z_{i,x} 
\ee
with $ L_x $ denoting isomorphic copy of the operator $L$ defined in \eqref{genera.1} and $G_{xy},q_{i,x}$ interaction functions;
let
\[\rho\equiv\sum_{x\in\ZZ^d} \varepsilon_x  \rho_x \qquad \mbox{where} \qquad \rho_x\equiv\bar\rho(q_x,p_x,u_x)\]
and assume 
\begin{eqnarray}  \label{G}
\sum_{x} \varepsilon_x G_{xy} \leq G \varepsilon_y, \qquad 
\sup_{x}\sum_{y }   |G_{xy}| \leq G  \label{G} \\
|q_{i,x}|\leq \tilde q\left((1+u_x^2)^{\frac\delta2}+(1+p_x^2)^{\frac\delta2}+(1+q_x^2)^{\frac\delta2}\right) \label{GG}
\end{eqnarray}
with some $\delta\in[0,1]$ and sufficiently small $G,\tilde q\in(0,\infty)$.
Then there exist $a,b\in(0,\infty)$ such that
\be \label{th} \mathcal{L}  \rho \leq -a \rho +b. \ee
Thus the corresponding model has an invariant measure for any $\lambda>0$.
\end{theorem}
\begin{proof}
We have
\begin{align*} 
\mathcal{L}  \rho &= 
\sum_{x\in \ZZ^d} \varepsilon_x   L_x \rho _x 
+ \sum_{y\neq x} \varepsilon_x G_{xy} u_y\pa_{u_x}  \rho_x 
 +\sum_{i=0,1,2}\sum_{x\in\ZZ^d} q_{i,x}  Z_{i,x}    \rho_x \\
 & \leq - \bar a\rho +\bar b \left(\sum_x \varepsilon_x\right) \\
 &+ \sum_{y\neq x\in\ZZ^d} \varepsilon_x G_{xy} u_y\left(2\bar C u_x+ g\bar H p_x \right) \\ 
 &+\sum_{ x\in\ZZ^d} \varepsilon_x q_{0,x}  \left(2\bar C u_x+ g\bar H p_x \right) \\
  &+\sum_{ x\in\ZZ^d} \varepsilon_x q_{1,x}   g\left( -2\bar C p_x- g\bar H u_x\right)\\
   &+\sum_{ x\in\ZZ^d} \varepsilon_x q_{2,x}   g\left(2\bar C q_x+ \bar R p_x -2g\bar C_xu_x - g^2\bar H p_x \right).  
\end{align*}
Therefore \eqref{th} follows using  \eqref{G} and \eqref{GG}.
\end{proof}
We note that
\begin{align*}   
&[Z_{0,x},\mathcal{L}] = -\lambda Z_{0,x} - Z_{1,x} +\sum_{y\in\ZZ^d} G_{yx} Z_{0,y} + \sum_{i=0,1,2}\sum_{y\in\ZZ^d} \left(Z_{0,x}q_{i,y}\right) Z_{i,y}\\
&\\
&[Z_{1,x},\mathcal{L}] = -Z_{2,x} + \sum_{i=0,1,2}\sum_{y\in\ZZ^d} \left(Z_{1,x}q_{i,y}\right) Z_{i,y}\\
&\\
&[Z_{2,x},\mathcal{L}] =(1+g^2)Z_{1,x} + \lambda g^2 Z_{0,y} + \sum_{i=0,1,2}\sum_{y\in\ZZ^d} \left(Z_{2,x}q_{i,y}\right) Z_{i,y}
-\sum_{y\in\ZZ^d} g^2 G_{yx} Z_{0,y} \,.
\end{align*}
Let us now look at the short time smoothing properties of the semigroup. In order to do so, we first calculate the following commutators:
\begin{align*}   
&[Z_{0,x},\mathcal{L}] = -\lambda Z_{0,x} - Z_{1,x} +\sum_{y\in\ZZ^d} G_{yx} Z_{0,y} + \sum_{i=0,1,2}\sum_{y\in\ZZ^d} \left(Z_{0,x}q_{i,y}\right) Z_{i,y}\, ,\\
&[Z_{1,x},\mathcal{L}] = -Z_{2,x} + \sum_{i=0,1,2}\sum_{y\in\ZZ^d} \left(Z_{1,x}q_{i,y}\right) Z_{i,y} \, ,\\
&[Z_{2,x},\mathcal{L}] =(1+g^2)Z_{1,x} + \lambda g^2 Z_{0,x} + \sum_{i=0,1,2}\sum_{y\in\ZZ^d} \left(Z_{2,x}q_{i,y}\right) Z_{i,y}
-\sum_{y\in\ZZ^d} g^2 G_{yx} Z_{0,y}\, ,
\end{align*}
and 
consider the following  functional 
$$
\Gamma_tf=\sum_{w\in \ZZ^d}\left[ \sum_{j=0}^2a_j t^{2j+1} \lv Z_{j,w} f\rv^2 + \sum_{j=1}^2  b_j t^{2j} (Z_{j-1,w}f) (Z_{j,w}f)
\right]+  a \lv f\rv^2, 
$$
for some strictly positive constants $a, a_j, b_j$ to be chosen.  Arguing as in the last section of Part I \cite{MV_I} 
to prove that $\Gamma_tf_t$ is a Lyapunov type functional, we  need to show that  
$$
\left(-\cl+ \partial_t \right)\left(\Gamma_tf_t \right) \leq 0,  \qquad f_t:= e^{t\cl }f_0, \,\,t\in (0,1). 
$$
To this end, using \cite[Lemma 2.2]{MV_I},  we note 
\begin{align}
\left(-\cl+ \partial_t \right)\left(\Gamma_tf_t \right)&= \sum_{w,y\in \ZZ^d}\left[ \sum_{j=0}^2 -2 a_j t^{2j+1} \lv Z_{0,y}Z_{j,w}
f_t\rv^2+ \sum_{j=1}^2 -2b_j t^{2j}(Z_{0,y}f_t) (Z_{j,w}f_t)\right] \label{q1}\\
& +\sum_{w\in\ZZ^d}\left[-2a\lv Z_{0,w}f_t\rv^2+ \sum_{j=0}^2 2 a_j t^{2j+1} (Z_{j,w}f_t)([Z_{j,w},\cl]f_t)\right.\label{q2}\\
&+ \left.\sum_{j=1}^2 b_j t^{2j}([Z_{j-1,w},\cl]f_t)(Z_{j,w}f_t)
+  b_jt^{2j}(Z_{j-1,w}f_t)([Z_{j,w}, \cl]f_t)   \right]\label{q3}\\
&+ \sum_{w\in\ZZ^d}\left[\sum_{j=0}^2 a_j (2j+1)t^{2j}\lv Z_{j,w}f_t\rv^2+
\sum_{j=1}^2 2j b_jt^{2j-1}(Z_{j-1,w}f_t)(Z_{j,w}f_t)\right] \label{q4}.
\end{align}
Let us set $(I):= \eqref{q1}, \, (II):=\eqref{q2}+\eqref{q3}, \,(III):=\eqref{q4}$ and look at these terms separately. 
Repeatedly using the quadratic Young's inequality, we get
$$
(I)\leq\sum_{y,w\in \ZZ^d}\sum_{j=0}^22(-a_j+b_{j+1}^2+1)t^{2j+1}\lv Z_{0,y}Z_{j,w}f_t\rv^2, $$ 
with the understanding that $b_3\equiv 0$. 
\begin{align*}
(II)&= 
\sum_{w\in\ZZ^d}\!
\left[\!-\lambda Z_{0,w}f_t - Z_{1,w}f_t +\!\sum_{y\in\ZZ^d}\!\left( \!
 G_{yw} Z_{0,y}f_t + \!
\sum_{i=0}^2 \left(Z_{0,w}q_{i,y}\right) Z_{i,y}f_t\right)\!
 \right] \left\{ 2a_0 t (Z_{0,w}f_t)+ b_1 t^2 (Z_{1,w} f_t) \right\}\\
 &+\left[-Z_{2,w}f_t +\sum_{y\in\ZZ^d} (Z_{1,w} q_{i,y})Z_{i,y}f_t \right]
 \left\{ 2a_1t^3 (Z_{1,w}f_t)+b_2t^4 (Z_{2,w}f_t)+b_1t^2(Z_{0,w}f_t)   \right\}
 -2a\lv Z_{0,w}f_t\rv^2\\
 &+\left[\!(1\!+\!g^2)Z_{1,w}f_t\! +\! \lambda g^2 Z_{0,w}f_t \!+
 \! \sum_{i=0}^2\!\sum_{y\in\ZZ^d} \!\left(Z_{2,w} q_{i,y}\right) Z_{i,y}f_t
  \!-\!\!\sum_{y\in\ZZ^d}\! g^2 G_{yw} Z_{0,y}f_t \! \right]
  \left\{ 2a_2 t^5 (Z_{2,w}f_t) \right. \\
 & \left. +b_2t^4 (Z_{1,w}f_t) \right\}.
\end{align*}
Now assume that
\be\label{asss1}
Z_{0,w}q_{2,y}=0, \qquad  \sup_{i,j=0,1,2}\,\sup_{y,w\in\ZZ^d}\lv Z_{j,y}q_{i,w}\rv \leq  \psi<1,
\qquad \mbox{for some constant } \psi>0
\ee
and that 
\be\label{asss2}
\sum_{x,y}|G_{xy}|<\bar{G}. 
\ee
We obtain for $t\in(0,1)$ 
\begin{align*}
(II)+(III)&\leq 
\sum_{w\in\ZZ^d}[-2a -2a_0\lambda t +6 (\lambda+\bar{G}+g+\psi+1)^2 
\sum_{j=0}^2(a_j^2+b_j^2)]
\lv Z_{0,w}f_t\rv^2\\
&+\sum_{w\in\ZZ^d}\left[ (\psi-1)b_1t^2 +6(\lambda+\bar{G}+ g+\psi+1)^2(a_1^2+a_2^2+b_2^2)\right] \lv Z_{1,w}f_t\rv^2\\
&+\sum_{w\in\ZZ^d}\left[  (\psi-1)b_2t^4+ 2b_2 t^5+6(\lambda+ g+\psi+1)^2a_2^2 \right] \lv Z_{2,w}f_t\rv^2.
\end{align*}
Choosing $a \gg a_0 \gg b_1  \gg a_1\gg b_2\gg a_2$, we can conclude the proof that the Lyapunov function $\Gamma_tf_t$ is non-increasing. We have therefore shown the following. 
\begin{theorem}
Consider the infinite dimensional dynamics $f_t:=e^{t\cl} f_0$ generated by the Markov operator \eqref{8_generator} and assume \eqref{asss1} and \eqref{asss2} hold. Then for any $j=0 \dd 2$ and any $w\in \ZZ^d$ we have
$$
\lv Z_{j,w} f_t\rv^2 \leq \frac{C}{t^{2j+1}} \|f_0\|_{\infty}^2.
$$
\end{theorem}

\label{LongTimeBehaviour} To consider the long time behaviour, we note first that for $\tilde B \equiv B -\lambda D$, $\lambda\in(0,\infty)$,
there exists vector fields $V_0$ and $V_{\pm}$ with constant coefficients such that
\[ [\tilde B,V_0]=\xi_0 V_0,\qquad [\tilde B,V_{\pm}]=\xi_{\pm} V_{\pm}
\]
with $\bar\xi_+=\xi_-$ and $\xi_0, \Re\xi_{\pm}\in(0,\infty)$.
Setting $V_{0,x}$ and $V_{\pm,x}$, $x\in\ZZ^d$, for isomorphic copies of $V_0$ and $V_{\pm}$, we consider the following generator
\be \label{8a_generator}
\mathcal{L} \equiv \sum_{x\in\ZZ^d}   L_x + \sum_{y\neq x\in\ZZ^d} G_{xy} u_y\pa_{u_x} +\sum_{x\in\ZZ^d} \eta_{0,x}  V_{0,x}\, , 
\ee
where $\eta_{0,x}$ are smooth bounded functions with bounded derivatives.
We have
\begin{align*} &\partial_s e^{ms}P_{t-s}\left(\sum_{x\in\ZZ^d} \left(|V_{0,x}f_s |^2 +|V_{+,x}f_s|^2 +|V_{-,x}f_s|^2\right)\right)= \\
&-2e^{ms}P_{t-s}\left(\sum_{x\in\ZZ^d} \left(|Z_0V_{0,x}f_s |^2 +|Z_0V_{+,x}f_s|^2 +|Z_0V_{-,x}f_s|^2\right)\right)
\\
&+ e^{ms}P_{t-s}\left(\sum_{x\in\ZZ^d} \left( (m-2\xi_0)|V_{0,x}f_s |^2 +(m-2\Re\xi_+)|V_{+,x}f_s|^2 +(m-2\Re\xi_-)|V_{-,x}f_s|^2\right)\right) \\
&
+  e^{ms}P_{t-s}\left(\sum_{x\in\ZZ^d} \left(
2 \sum_{y\in\ZZ^d}    (V_{0,x}\eta_{0,y}) V_{0,x}f_s \cdot 
 V_{0,y}f_s \right)\right)\, . 
\end{align*}
Hence, if
\be\label{assm}   m \leq \min\left( 2\xi_0 - \sup_{z\in\ZZ^d}\sum_{y\in\ZZ^d}   \frac12\left( \|V_{0,z}\eta_{0,y}\|_\infty +\|V_{0,y}\eta_{0,z}\|_\infty\right)   , 2\Re\xi_\pm \right),
\ee
then 
$$
\partial_s e^{ms}P_{t-s}\left(\sum_{x\in\ZZ^d} \left(|V_{0,x}f_s |^2 +|V_{+,x}f_s|^2 +|V_{-,x}f_s|^2\right)\right) \leq 0.
$$
Integrating the above expression in $[0,t]$ proves the statement of Theorem \ref{thm23} below.
\begin{theorem}\label{thm23}
With the notation introduced so far and assuming \eqref{assm},  the following exponential decay holds along the semigroup generated by $\cl$:
\[\sum_{x\in\ZZ^d} \left(|V_{0,x}f_t |^2 +|V_{+,x}f_t|^2 +
|V_{-,x}f_t|^2 \right)
 \leq e^{-mt} P_t \left(\sum_{x\in\ZZ^d} \left(    
|V_{0,x}f_0 |^2 +|V_{+,x}f_0|^2 +|V_{-,x}f_0|^2\right) \right).
\]
\end{theorem} 
\newpage

\section{Lie groups of Heisenberg type}\label{Subs: Lie groups of Heisenberg type} \label{Sec3}
In this section, we present an application our results to Markov generators on groups of Heisenberg type. For a detailed introduction to such groups we refer the reader to \cite{BLU}.
\begin{definition}
Let $\mathfrak{g}$ be a Lie algebra whose centre is \(\mathfrak{z}\) and let  \( \mathfrak{v}:=\mathfrak{z}^{\perp} \).We say that $\mathfrak{g}$ is of Heisenberg-type (or simply $H$-type) if \[[\mathfrak{v},\mathfrak{v} ]= \mathfrak{z}\] and there exists an inner product $\langle \cdot, \cdot\rangle$ on $\mathfrak{g}$ with $\left<\mathfrak{z}, \mathfrak{v}\right>=0$ such that  for any $Z\in\mathfrak{z}$, the map $J_Z: \mathfrak{v}\mapsto\mathfrak{v}$ given by
\[
\langle J_ZX, Y\rangle = \langle [X, Y], Z\rangle,
\]
for $X, Y\in \mathfrak{v}$, is an orthogonal map whenever \(\left< Z,Z\right>=1\).  An $H$-type group is a connected and simply connected Lie group $\mathbb{G}$ whose Lie algebra is of $H$-type.
\end{definition}

 Consider an $H$-type group $\mathbb{G}$ which is isomorphic to \(\R^n = \R^{m+r}\). We write elements of \(\mathbb{G}\) as \(w=(x,z)\), with  \(x \in \R^m\) and \(z \in \R^r\). We moreover denote the left-invariant fields by \(Z_{1},\dots, Z_m\) and their commutators by \(Z_{m+1}=[Z_1,Z_2],\dots     ,Z_{m+r}=[Z_{m-1},Z_m]. \) The Lie algebra is naturally equipped with a first order  operator \(D\) which generates dilations  and satisfies  
\begin{equation}
\label{Condition_onD}
e^{sD} Z_k e^{-sD}=e^{sl_k}Z_k \textrm{ and }\;[Z_k,D]=l_k Z_k ,
\end{equation}
 for all \(k=1,\dots,n\) and \(s>0\), where \(l_k=1\) for \(k=1, \dots m\) and \(l_k=2\) otherwise (the constants \(l_{k} \) reflect the layer of the Lie algebra that \(Z_{k} \) belongs to).
More specifically,  \(D\) is given as the generator of the dilations \(\delta_\lambda(w)=(\lambda x,\lambda^2 z),\) by
\begin{align}\label{Dcomputation}
D &=\partial_\lambda  \mid_{_{\lambda=1}}\delta_\lambda(w)=x\cdot \nabla _m+ 2z  \cdot \nabla_r ,
\end{align}
where \(\nabla_m\) and \(\nabla_r\) denote the Euclidean gradients on \(\R^{m}\) and \(\R^r\),  respectively.

We consider an operator \(\mathcal{L}\) on  $\mathbb{G}^{\ZZ^d}$, which is obtained as the infinite-dimensional limit, as \(\La \uparrow \ZZ^d\), of
\begin{align}
\lla:= & \sum_{x \in \ZZ^d}\mathcal{L}_x+ \sum_{y \in \Lambda}\sum_{r=1}^nq_{r} \sigma_{y} Z_{r,y}
\\ := &  \sum_{x \in \ZZ^d}  \left(\sum_{i,j=1}^{m}(\delta_{ij}+G_{ij}) Z_{i,x}Z_{j,x} +\sum _{i=1}^n p_{i}Z_{i,x}- \delta D_x\ \right)+\sum_{y \in \La}  \sum_{r=1}^nq_{r} \sigma_{y} Z_{r,y},
\end{align}
where \(\delta ,p_{i}, q_r >0\), \(G=(G_{ij})_{i,j=1}^m\) satisfies \(G+I>0\)  and the family \( \{\sigma_y \}_{y \in \La}\) is such that the quantities \( \Vert \sigma_y \Vert_\infty, \Vert Z_{r,y}\sigma_{\tilde{y}} \Vert _\infty\) are uniformly bounded in \(y,\tilde{y} \in \ZZ^d\) for \(r=1, \dots , n\). Similar operators were considered in   \cite{DragoniKontisZegarlinski2011}, under the assumption that \(\delta\) is large enough. Here, however, we only require that \(\delta>0\).
\par In what follows, we construct a Lyapunov function for \(\mathcal{L}\), which satisfies the assumptions of Section 4 of \cite{MV_I} .
The group $\mathbb{G}$ is equipped with the so-called   Folland-Kaplan gauge \(N\), defined by 
\begin{align}\label{gaugedef} N(w)= \left( \vert x\vert^{4}+16\vert z  \vert ^2\right)^{1/4},\end{align}
where  \(\vert \cdot \vert \) denotes the Euclidean norm. A computation shows \cite{DGR}  that the sub-gradient and the sub-Laplacian of \(N\) read 
\begin{equation}\label{Ngradbd}
 \vert \nabla _0 N \vert ^2:=\sum_{i=1}^m \vert Z_i N \vert ^2(w)= \frac{\vert x\vert ^2}{N^2(w)}
\end{equation}
 and
\begin{equation}  
\Delta _0 N := \sum_{i=1}^{m} Z_i^2N(w)=3\frac{\vert x \vert ^2}{N^3(w)},
\end{equation}
respectively. In particular, \(\vert \nabla_0 N\vert  \le 1\), since \(\vert x \vert \le N\). Moreover,  for \(j >m\),  \(Z_{j}= \partial_j\), for all \(i=1, \dots, r\) \cite{BLU}, and therefore  
\[Z_{m+i} N=\frac{8 z_{i}}{N^{3}},\]
which in turn implies
\[ \sum_{i=1}^n \vert Z_i N \vert ^2(w)=\frac{\vert x\vert ^2}{N^2(w)} +\frac{64 \vert z \vert ^2}{N^6 (w)} \le1+\frac{4}{N^{2}(w)},\]
using that \(\vert x \vert \le N(w)\) and \(16 \vert z \vert ^2 \le N^4(w).\)
Let  \((\textup{Hess} N)^*\) denote the symmetrised Hessian  of \(N\), i.e. the \(m \times m\) matrix with elements 
 \[(\textup{Hess} N)^*_{ij}= \frac{1}{2}(Z_i Z_j N + Z_j Z_i N),\]
for \(i,j=1, \dots, m\).  It was shown in    \cite{garofalotyson} that 
\[(\textup{Hess} N)^*_{ij}=\frac{1}{N^7} \left( N^4 \vert x \vert^2  \delta_{ij}+2N^4 \left(x_i x_j + \sum_{s=1}^r B_{is}B_{js} \right) - 3 \left<A,e_i \right>\left<A,e_j \right> \right),\]
where \(e_i\) denotes the \(i^{th}\) element of the standard basis of \(\R^m\) and  for \(s=1, \dots, r, \)
\[B_{is}= B_{is}(w)= \left< J_{e_{m+s}}x,e_i \right>\]
and \(A= \vert x \vert ^2 x+ 4 J_z x. \)  Using the identities  \( \vert J_z x \vert = \vert x \vert \vert z \vert\) and \(\left< J_zx,x \right>=0\) (see e.g. \cite{BLU}), we see  that    \(B_{is} \le \vert J_{e_{m+s}} x \vert  = \vert x \vert \)   and \(\vert A \vert ^2 = \vert x \vert ^4 \vert x \vert ^2+ 16 \vert x \vert ^2 \vert z \vert ^2= N^4 \vert x \vert^2.   \) Using  Young's inequality, we thus arrive at the estimate
\begin{align*}\left \vert (\textup{Hess}N)^*_{ij} \right \vert &\le \frac{1}{N^7} \left( N^4 \vert x \vert^2  \delta_{ij}+2N^4\left(  \frac{x_i^2+x_j^2}{2} +r \vert x \vert ^2\right)+3 \vert A \vert ^2\right)  
\\ & \le \frac{1}{N^7} \left( N^4 \vert x \vert^2  \delta_{ij}+2(1+r)N^4\vert x \vert ^2+3 N^4 \vert x \vert^2\right)  
\\ & \le \frac{\delta_{ij}+2r+5}{N}, \end{align*} 
where we used once again that $\vert x \vert  \le N$.
Summing over \(i, j\), we obtain 
\[ \sum_{i,j=1}^m \left \vert (\textup{Hess}N)^*_{ij} \right\vert  \le \frac{m+2rm^{2}+5m^{2}}{N}.\]
Let us also observe that  
\begin{align*}
DN= \sum_{i=1}^m x_i \partial_{i}N+ 2\sum_{i=1}^rz_i \partial_{m+i}N= N^{-3}\left( \vert x \vert ^4+16 \vert z \vert ^2\right)=N.
\end{align*}
Combining the above estimates and using the Cauchy-Schwarz inequality,  we conclude that for the operator \(L := \sum_{i,j=1}^{m}(\delta_{ij}+G_{ij}) Z_{i}Z_{j} +\sum _{i=1}^n p_{i}Z_{i}- \delta D \) we have   
\begin{align*} L N&= \sum_{i,j=1}^m (\delta_{ij}+G_{ij}) (Z_{i}Z_{j} N) +\sum _{i=1}^n p_{i}Z_{i}N-\delta DN
\\ & \le  \max_{i,j}(\delta_{ij}+G_{ij} ) \sum_{i,j=1}^m \left \vert (\textup{Hess} N)^*_{ij} \right\vert  +\sqrt{\sum_{i=1}^np_i^2}  \sqrt{\sum_{i=1}^n \vert Z_i N \vert ^2}-\delta DN
\\ & \le \frac{c_{1}}{N} +\vert p \vert \sqrt{1+ \frac{4}{N^{2}}}-\delta N,
\end{align*}
using the Cauchy-Schwarz inequality,    where \(c_1=\max_{i,j}(\delta_{ij}+G_{ij} )( m+2rm^{2}+5m^{2})\) and \( \vert p \vert^2 = \sum_{i=1}^n p_i^2\). Let\begin{align}\label{Wdef}W(w)=\sqrt{1+N(w)^2}.\end{align} Then \(W \ge 1 \) is a smooth function. We claim that there exist constants \(C_{1}\ge 0\) and \(C_2>0\) such that \begin{displaymath}
LW<C_1-C_2W.
\end{displaymath}Indeed, as L is a diffusion generator,
we have
\begin{align*}
LW&= \frac{N}{\sqrt{1+N^2}} L N  + \frac{1}{({1+N^2})^{3/2}}\sum_{i,j=1}^{m}(\delta_{ij}+G_{ij}) (Z_{i}N)(Z_{j} N)
\\ & \le \frac{c_{1}   +\vert p \vert \sqrt{N^2+4}-\delta N^2}{\sqrt{1+N^2}}+c_2 \\ & \le c_1+2\vert p \vert+\delta+c_{2}-\delta W,
\end{align*} 
with \(c_1 \) as above and  \begin{displaymath}
c_{2}=\max_{i,j}(\delta_{ij}+G_{ij}),
\end{displaymath}and using the elementary inequality \[-\frac{x^2}{\sqrt{1+x^2}} \le -\sqrt{1+x^2}+1.\]
Finally, by considering isomorphic copies of \(W \) at different points of the lattice and defining \(W_\Lambda = \sum_{i \in \Lambda}W_i\) we obtain a suitable Lyapunov function for \(\mathcal{L}_\Lambda\).

\section{Exponential Decay for Filiform Algebra with Full Dilation.}
\label{ExpDecay4FilliformAlgebra_FullDilation}
Let  $D$ and $\{Y_i\}_{i=0\dd N+1}$, be first order differential operators on $\R^m$ and set
\[
Y_i:=
\begin{cases}
B, \qquad \text{for } i=0 \nonumber\\
Z_{i-1}, \qquad \text{for } i=1,..,N+1 . \nonumber
\end{cases}
\]
Consider the following second order differential
operator on $\mathbb{R}^{m}$ 
\[
L=Y_{1}^{2}+Y_0 -\lambda D, \quad \lambda>0.
\]
 Assume
that for some $N\in\mathbb{N}$ , $N\geq1$, there exist $Y_{1},\dots,Y_{N+1}$
such that the following commutator relations hold true:
\begin{align*}
&[Y_0,Y_{j}]=Y_{j+1}\qquad j=1,\dots,N,
\qquad 
[Y_0,Y_{N+1}]=0, \\
&  [Y_{i},Y_{j}]=0 \qquad  1\leq i,j\leq N+1 \qquad \mbox{and}\\
& [Y_i,D]=\kappa_i Y_i,\qquad  i=0,\dots ,N+1\,,
\end{align*}
for some $\kappa_i\in(0,\infty)$.
For convenience of notation we will set $Y_i\equiv 0$ for $i>N+1$.
Hence we have
\begin{equation} \label{CR_L}
[Y_i,L]=
2 \, \delta_{i,0}
Y_1Y_2  
-(1-\delta_{i,0})Y_{i+1}-\lambda\kappa_i Y_i.
\end{equation}
For $n\in\mathbb{N}$, let
\[\Y_{\kk,n}\equiv Y_{k_1}...Y_{k_n}\]
where $\kk\equiv(k_1,...,k_n)$,  $k_i\in\mathbb{N}\cup\{0\}$. If $\e_i$ is the $i$-th vector of the standard basis of $\R^n$, 
we have
\begin{align*}
[\Y_{\kk,n}, L]=& 
-\lambda \,\kappa_\kk\Y_{\kk,n} - \sum_{\kk}\sum_{i=1}^n (1-\delta_{k_i,0})\Y_{\kk+\e_i}
\\
&+\sum_{j=1}^n\sum_\kk  \delta_{k_j,0} Y_{k_1}{\dots}  Y_{k_{j-1}} Y_{1}Y_{2} {\dots} Y_{k_{j+1}}Y_{k_n}\, ,
\end{align*}
with the convention that in the last sum for $i=1$ and $i=n$ there are no factors on the left and right of $Y_1Y_2$, respectively. 
As before, for a smooth
bounded function $f$, we set $f_t\equiv e^{tL}f$, $t\in(0,\infty)$.
Let
$m_{j}\in(0,\infty)$,
$j\in\mathbb{N}$, be such that 
\begin{equation} \label{nus}
m_0=m_2,\, m_1=1\, \text{ and }\,  m_j\geq m_{j+1},  \text{ for }  j\geq 1 .
\end{equation}
We also set $m_{\kk}\equiv \sum_j m_{k_j}$ and $\kappa_{\kk}\equiv \sum_j \kappa_{k_j}$ and define 
\begin{equation}
\widetilde\Gamma_{t}^{(n)}(f_t)\equiv  \sum_{\kk}  e^{m_{\kk}t}|\Y_{\kk ,n}f_t|^{2}\,.
\label{Aii.0} 
\end{equation}

\begin{theorem}\label{pro_Aii.1}  
With the notation introduced in this section, for any $n\in\mathbb{N}$, there exist $\lambda\in(0,\infty)$ such that
\be\label{gtilde}
\widetilde\Gamma_{t}^{(n)}(f_t) \leq \widetilde\Gamma_{0}^{(n)}(f), \qquad \mbox{for all } t>0.
\ee
Therefore, for all $\kk, n$ and $t>0$ we have
$$
\lv \Y_{\kk,n} f_t\rv^2\leq e^{-m_{\kk}t} \,\widetilde\Gamma_{0}^{(n)}(f).
$$
\end{theorem}
\begin{proof} Given $n\in\NN$ and simplifying notation $\Y_\kk\equiv\Y_{\kk,n}$,
we have
\begin{align*}
 (-L+\partial_{t})\widetilde\Gamma_{t}^{(n)}(f)  =&   -2\sum_\kk e^{m_\kk t}|Y_{1}\Y_\kk f_t|^{2} 
+2\sum_\kk e^{m_\kk t} \Y_\kk f_t \cdot [\Y_\kk ,L] f_t    
 +\sum_\kk m_\kk  e^{m_\kk t}|\Y_\kk f_t|^{2} \\ 
 =&-2\sum_\kk e^{m_\kk t}|Y_{1}\Y_\kk f_t|^{2}   \nonumber +\sum_\kk (m_\kk - 2\lambda \kappa_\kk) e^{m_\kk t}|\Y_\kk f_t|^{2} \\ & -2\sum_{i=1}^n\sum_\kk (1-\delta_{k_i,0}) e^{m_\kk t} \Y_\kk f_t\cdot
\Y_{\kk +\e_i}f_t \label{Aii.l3}  \nonumber \\ 
& -4\sum_{i=1}^n\sum_\kk  \delta_{k_i,0}  e^{m_\kk t} \Y_\kk f_t\cdot (Y_{k_1}..Y_{k_{i-1}} Y_{1}Y_{2} Y_{k_{i+1}}..Y_{k_n}f_t ),
\end{align*}
recalling the  convention that in the last sum for $i=1$ and $i=n$ there is no factors on the left and right of $Y_1 Y_2$, respectively.
Since
\begin{align*}
Y_{k_1}\dots Y_{k_{i-1}} Y_{1}Y_{2} Y_{k_{i+1}}\dots Y_{k_n} =&\;
Y_{1}Y_{k_1}\dots Y_{k_{i-1}} Y_{2} Y_{k_{i+1}}\dots Y_{k_n} \\&+ \sum_{l=1}^{i-1}
Y_{k_1}\dots[Y_{k_l},Y_1]\dots Y_{k_{i-1}} Y_{2} Y_{k_{i+1}}\dots Y_{k_n},
\end{align*}
with the first term of order $n+1$ and each term in the sum of order $n$,
taking into the account that $m_0=m_2$, we have the following bound
\begin{align*}
& 4\lv \sum_{i=1}^n\sum_\kk  \delta_{k_i,0}  e^{m_\kk t} \Y_\kk f_t\cdot (Y_{k_1}..Y_{k_{i-1}} Y_{1}Y_{2} Y_{k_{i+1}}..Y_{k_n}f_t) \rv \\
& \qquad  \leq 
2\sum_\kk e^{m_\kk t} |Y_{1}\Y_\kk f_t|^{2} +
n(n-1)(n-2) \sum_\kk e^{m_\kk t} |\Y_\kk f_t|^{2}.    
\end{align*}
We also note that, because $m_{\kk}\leq m_{\kk+\e_i}$ whenever $k_i\neq N+1$,
we have
\begin{equation}
2\left|\sum_{i=1}^n\sum_\kk (1-\delta_{k_i,0}) e^{m_\kk t} \Y_\kk f_t\cdot
\Y_{\kk +\e_i}f_t \right| \leq 2 n \sum_\kk  e^{m_\kk t} |\Y_\kk f_t|^2 .\nonumber
\end{equation}
Hence we get
\begin{equation*}
  (-L+\partial_{t})\widetilde\Gamma_{t}^{(n)}(f) \leq 
\sum_\kk (m_\kk - 2\lambda \kappa_\kk+2n + n(n-1)(n-2)) e^{m_\kk t}|\Y_\kk f_t|^{2}.
\end{equation*}
That is, if
\[ \max_\kk \left( m_\kk - 2\lambda \kappa_\kk+2n + n(n-1)(n-2) \right) \leq 0
\]
we obtain
\[   (-L+\partial_{t})\widetilde\Gamma_{t}^{(n)}(f_t) \leq 0.
\]
\end{proof}

\newcommand{\fpps}{\frac{\partial}{\partial s}}
\newcommand{\cE}{\mathcal{E}}

\section{Filiform Algebras with Partial Dilation.}
\label{Filliform Algebras Partial Dilation}
Consider the following fields in $\mathbb{R}^n$
\begin{align*}
Z_0\equiv \partial_1, \; B\equiv \partial_2 + x_1\partial_3+x_3\partial_4+..+x_{n-1}\partial_n,
 \end{align*}
 For \(j=0,\dots, n-1,\) we set 
 \begin{align*}
 Z_{j+1} \equiv  [B,Z_j]= (-1)^{j+1}\partial_{j+1},
 \end{align*}
 and note that  \begin{align*}
[B,Z_n]=0 \text{ and } [Z_i,Z_j] = 0 \text{ for all } i,j=0,\dots,n .\end{align*}
 Next we introduce a partial dilation
 \begin{align*}
& D_0\equiv x_1\partial_1.
\end{align*}
Observe that \begin{align*}
 [Z_0,D_0]&=Z_0,\\ [Z_i,D_0]&= 0, \text{ for } i\neq 0\\
[B,D_0]& = -x_1\partial_3,
\\ \underbrace{[..[B,D_0],..,D_0]}_{n-times}&= (-1)^n x_1\partial_3.
 \end{align*}
We consider the  generator
given by \[
 L\equiv Z_0^2 +B-\lambda D_0
 \]
 with $\lambda \in [0,\infty)$. Let $f_t\equiv P_tf\equiv e^{tL}f$.
We note that  the following  vector
\[V\equiv \sum_{j=0}^n \left(\lambda^{n-j}\chi_{\{\lambda< 1\}}+\frac1n\chi_{\{\lambda = 1\}}
+ \lambda^{j-n}\chi_{\{\lambda > 1\}}\right) Z_j\] 
 satisfies
 \[ [L,V]=[B-\lambda D_0,V]= \lambda V.
 \]
Define
\be \label{8F_generator}
\mathcal{L} \equiv \sum_{\rw\in\ZZ^d}   L_\rw  
 +\sum_{\rw\in\ZZ^d} q_{\rw}  V_{\rw} 
\ee 
 with $L_\rw$ and $V_{\rw} $ denoting the isomorphic copy of $L$ and $V$, respectively,
with $q_{\rw} $ being bounded smooth functions with bounded derivatives dependent on many variables.
Since
 \begin{align*}
 2\sum_{\rw\in\ZZ^d} V_{\rw}f_s \cdot [V_{\rw}, \mathcal{L} ]f_s=& -2\lambda \sum_{\rw\in\ZZ^d} |V_{\rw}f_s|^2
 +  2\sum_{\rw\in\ZZ^d}  V_{\rw}f_s \cdot \sum_{\ry\in\ZZ^d} 
 (V_{\rw} q_{\ry} )V_{\ry}  f_s \\
 \leq& \;2(\eta - \lambda ) \sum_{\rw\in\ZZ^d} |V_{\rw}f_s|^2,
 \end{align*}
 with 
 \[ \eta \equiv \frac12\sup_{\rw\in\ZZ^d}\sum_{\ry\in\ZZ^d} \left(\| V_{\rw} q_{\ry} \|_\infty+ \|V_{\ry} q_{\rw} \|_\infty\right),
 \]
 we have
 \[\sum_{\rw\in\ZZ^d} |V_{\rw}f_t|^2 \leq e^{2(\eta-\lambda)t} P_t \sum_{\rw\in\ZZ^d} |V_{\rw}f|^2.\]
 By similar arguments one has
 \[\sum_{\rw_1,..\rw_k\in\ZZ^d} \left| V_{\rw_1}..V_{\rw_k}f_t\right|^2 \leq e^{2k(\eta-\lambda)t} P_t \sum_{\rw_1,..\rw_k\in\ZZ^d} \left|V_{\rw_1}..V_{\rw_k}f\right|^2.
 \] 
 The short time estimates can be done as in a similar way as in other examples.
 \section{Heisenberg model with partial dilation.}
  \label{H_model_PartialDilation} 
In this section we provide an explicit example of dissipative dynamics with partial dilation which forces exponential concentration along a suitable vector field.  Let
\begin{align*}
Z_0&\equiv X\equiv\partial_x+\frac12y\partial_z,\\  B&\equiv Y\equiv \partial_y -\frac12x\partial_z, \\
Z_1 &\equiv [B,Z_0]= \partial_z,
\end{align*}
so that \(
[B,Z_1]=[Z_0,Z_1]=0\).
 Next, we introduce a partial dilation
\[
  D_0\equiv x \partial_x\]
such that \(
  [\partial_x,D_0]=\partial_x\) and \(  [Z_1,D_0]= 0\),
 and the following generator
 \[
 L\equiv X^2 +\xi Y-\lambda D_0,
 \]
 with $\lambda \in [0,\infty)$. Let $f_t\equiv P_tf\equiv e^{tL}f$.
 We note that
\[
[\partial_x, L]= -\frac\xi2 \partial_z -\lambda \partial_x 
\]
and \(  [Z_1,L]=0\).
Thus, setting  \[V \equiv \kappa  \partial_x  +
 \kappa^{-1}\partial_z \]
  with $\kappa\equiv \sqrt{\frac{2\lambda}{\xi}}$, 
   we have
 \[ [V,L] = -\lambda V.\]
 This model has no invariant probability measure, but it has the following concentration property. 
 
 \begin{theorem}
 For any $\lambda, \xi > 0$, one has
 \be  
 \left(V f_t \right)^2
  \leq e^{ -2 \lambda t} P_t\left(V f  \right)^2
 \ee
 for any $f$ for which the right hand side is well defined.
  \end{theorem}
 
\noindent{\bf  Partial Concentration for infinite dimensional model.}\\
   \label{Concentration4infdim}
   
\noindent   We consider the following model

   \[\mathcal{L} \equiv \sum_i \mathcal{L}_i +\sum_i \left( q_i V_i  + \gamma_iX_i + \eta_i Z_i \right)\]
  with summation over $i\in\mathbb{Z}^d$, 
  \[\mathcal{L}_i \equiv X_i^2 +\xi Y_i -\lambda D_i^{(0)}\]
  where $V_i,X_i, Y_i, Z_i, D_i^{(0)}$ are copies of operators 
  $V,X,Y,Z, D_0$ defined above, and $\xi,\lambda$ are positive constants and$q_i, \gamma_i , \eta_i$ are smooth functions for which the infinite dimensional semigroup is well defined on the graph $\mathbb{Z}^d$.
We have the following result.
  \begin{theorem}
 Suppose $\gamma_i , \eta_i$ depend only on $y_j$ , $j\in\mathbb{Z}^d$.
 For any $\lambda, \xi > 0$, if
 \[0\leq m\leq 2 \lambda - \sup_k\sum_{j} \left(\left\|V_j q_k\right\|_\infty + \left\|V_kq_j\right\|_\infty \right) , 
 \]
 then
 \be  
 \sum_j |V_jP_t f|^2 
  \leq e^{ - m t} P_t\sum_j |V_jf|^2 
 \ee
 for any $f$ for which the right hand side is well defined.
  \end{theorem}
  
  \begin{proof}
  We have
\begin{align*}
 \partial_s e^{ms}P_{t-s} \sum_j |V_jf_s|^2 
 =& -2 e^{ms}P_{t-s} \sum_{j,k} |X_kV_jf_s|^2\\
&  +e^{ms}P_{t-s} (m - 2 \lambda  ) \sum_{j}  | V_jf_s|^2  
 -2   e^{ms}P_{t-s} \sum_{j,k} \left(V_j q_k\right) V_jf_s\cdot  V_kf_s\\
 \end{align*}
 and so
\[ \partial_s e^{ms}P_{t-s} \sum_j |V_jf_s|^2 
 \leq 
 - C e^{ms}P_{t-s}\sum_{j}  | V_jf_s|^2  
 \]
 with
 \[ C \equiv  2 \lambda - m - \sup_k\sum_{j} \left(\left\|V_j q_k\right\|_\infty + \left\|V_kq_j\right\|_\infty \right)\]
 If $C \geq 0$, the statement follows.
  \end{proof}


\section{B-S Model with Interaction.} \label{B-S Model}  
 We consider the model in $\mathbb{R}$ with the following generator \[
 L= x^2\partial_x^2 +\varepsilon\partial_x -\lambda x\pa_x\, , \qquad \lambda\geq 0,\,  \varepsilon\in\mathbb{R}.
\]
Setting $Z_0:= x\pa_x$ we note that
\[x^2\partial_x^2 = (x\partial_x)^2 -x\partial_x \equiv Z_0^2- Z_0 , 
\] 
so the generator can be rewritten as 
\[L=Z_0^2 +B -(\lambda+1) Z_0  , \]
where  $B:=\varepsilon\partial_x$.   Also, 
\[Z_1:= [B,Z_0]=[L,Z_0]=\varepsilon\partial_x=B.  \]
We note that for $\langle x\rangle\equiv \left(1+x^2\right)^{\frac12}$, we have \label{BS.Compact_findim}
\[ L \langle x\rangle \leq  \varepsilon 
- \lambda  \langle x\rangle ,  \]
hence for $\lambda>0$, the semigroup $e^{tL}$ has an invariant probability measure. 
We  also have the following gradient type bound\begin{versionb} \label{expDec_finDim}
\begin{eqnarray*}
&\partial_sP_{t-s}\left(e^{ms}|\partial_xf_s|^2\right)=\\
&P_{t-s}\left(-2e^{ms}|Z_0\partial_xf_s|^2 +2e^{ms}\partial_xf_s\cdot( \{Z_0,\partial_x\})f_s -2(\lambda+1)e^{ms}|\partial_xf_s|^2+m e^{ms}|\partial_xf_s|^2 \right)\\
&=P_{t-s}\left(-2e^{ms}|Z_0\partial_xf_s|^2 +2e^{ms}\partial_xf_s\cdot 2 Z_0\partial_x f_s  -(2 \lambda -m) e^{ms}|\partial_xf_s|^2 \right)\\
&=P_{t-s}\left(-2e^{ms}\left(Z_0\partial_xf_s -\partial_xf_s\right)^2    -(2 (\lambda -1)-m) e^{ms}|\partial_xf_s|^2 \right)\\
\end{eqnarray*}
and hence
\end{versionb}
\begin{equation*}
 \partial_sP_{t-s}\left( |\partial_xf_s|^2\right) 
 =P_{t-s}\left(-2 \left(Z_0\partial_xf_s -\partial_xf_s\right)^2    - 2 (\lambda -1)  |\partial_xf_s|^2 \right),\\
\end{equation*}
\noindent
which in turn implies 
\[|\partial_xf_t|^2\leq e^{-2(\lambda-1)t}P_t|\partial_xf |^2.
\]
More generally, for any $n\in\mathbb{N}$ we have 
\begin{versiona}
\begin{eqnarray*}
\partial_sP_{t-s}\left( |\partial_x^nf_s|^2\right) 
&=P_{t-s}\left(-2  |Z_0\partial_x^n f_s|^2  
 +2 \partial_x^nf_s \left(2nZ_0\partial_x^n f_s - n \lambda\partial_x^n f_s \right)  \right)\\
&= P_{t-s}\left(-2  (Z_0\partial_x^n f_s - n\partial_x^nf_s)^2
 - 2n(\lambda-n)|\partial_x^nf_s|^2  \right)
\end{eqnarray*}
from which we deduce
\end{versiona}
\begin{versiona}
\begin{equation*}
  |\partial_x^nf_t|^2 \leq
 e^{-2n(\lambda-n)t}P_t|\partial_x^n f |^2 . 
\end{equation*}
\end{versiona}
\noindent
In particular, if $\lambda \in(1,\infty)$ we have pointwise decay to equilibrium for differentiable functions.

\noindent
To study the smoothing estimates we  recall the notation
$$
\{X,Y\}= XY+YX=2XY-[X,Y],
$$
for any two given  differential operators $X$ and $Y$, and  we
and introduce the Lyapunov functional
$$
\Gamma_t (f_t)= at|Z_0f_t|^2+bt^3|Z_1 f_t|^2+ cs^2Z_0f_t\cdot Z_1  f_t +d f_t^2,
$$
for some positive real constants $a,b,c$ and $d$. Using 
\[
[Z_1, L]=\{Z_0, Z_1\}- (\lambda+1)Z_1= 2Z_0Z_1-\lambda Z_1
\]
and \cite[Lemma 2.2]{MV_I}, we  have 
\begin{align*}
\partial_sP_{t-s}\left( \Gamma_s f_s \right)=&P_{t-s}\left[(-L+\pa_s)\Gamma_s f_s \right]\\
=&P_{t-s}\left(-2as|Z_0^2 f_s|^2 -2bs^3|Z_0Z_1  f_s|^2  
-2cs^2Z_0^2f_s\cdot Z_0Z_1  f_s -2d |Z_0f_s|^2 \right)\\
&+P_{t-s}\left(  - 2as Z_0 f_s\cdot  Z_1 f_s 
+2 bs^3 Z_1 f_s\cdot (\{Z_0, Z_1 \} -(\lambda+1)Z_1 )f_s \right)\\
&+P_{t-s}\left( -  cs^2 |Z_1 f_s|^2 +  cs^2 Z_0f_s\cdot  (\{Z_0,Z_1 \} -
(\lambda+1)\partial_x)f_s )
\right)\\
&+P_{t-s}\left( a |Z_0f_s|^2+3bs^2|Z_1  f_s|^2+ 2cs Z_0f_s
\cdot Z_1  f_s     \right).
\end{align*}
Therefore, by repeatedly using Young's inequality, we obtain
\begin{align*}  
\partial_sP_{t-s} \left( \Gamma_s f_s \right) \leq & P_{t-s}\left(-2as|Z_0^2 f_s|^2 -2bs^3|Z_0Z_1  f_s|^2  
-2cs^2Z_0^2f_s\cdot Z_0Z_1  f_s \right)\\
&+P_{t-s}\left(  -2as Z_0 f_s\cdot Z_1 f_s 
+2 bs^3 Z_1 f_s\cdot ( 2 Z_0 Z_1  - \lambda Z_1 )f_s \right)\\
&+P_{t-s}\left( - cs^2 |Z_1  f_s|^2 +  cs^2 Z_0f_s\cdot  ( 2 Z_0 Z_1   - \lambda Z_1 )f_s )
\right)\\
&+P_{t-s}\left( (a+\delta^{-1}c^2-2d) |Z_0f_s|^2+(3b+\delta)s^2|Z_1 f_s|^2   \right) \\ 
\leq &P_{t-s}\left(-2as|Z_0^2 f_s|^2 -2bs^3|Z_0Z_1  f_s|^2  
-2cs^2Z_0^2f_s\cdot Z_0Z_1  f_s \right)\\
&+P_{t-s}\left(  4\frac{a}{c}  |Z_0 f_s|^2 
+ \frac12 cs^2 |Z_1  f_s |^2
  + bs^3 |Z_0 Z_1 f_s|^2  - 4\lambda bs^3 |Z_1  f_s|^2\right)\\
&+P_{t-s}\bigg( -  cs^2 |Z_1  f_s|^2 
+  \frac{2c^2}{ b} s  |Z_0f_s|^2 +  \frac12 b s^3 |Z_0 Z_1  f_s|^2
\\ & \qquad \qquad + \frac12 \lambda \delta^{-1}c^2s  |Z_0f_s|^2 + \lambda \frac12\delta s^3 |Z_1  f_s|^2 
\bigg)\\
&+P_{t-s}\left( (a+\delta^{-1}c^2-2d) |Z_0f_s|^2+(3b+\delta)s^2|Z_1  f_s|^2   \right).
\end{align*}
Putting everything together we have 
\begin{align}
\pa_s P_{t-s}\left(\Gamma_s f_s \right) \leq & P_{t-s}\left(-2as|Z_0^2 f_s|^2 -2 bs^3|Z_0Z_1  f_s|^2  
+c^2s^2\lv Z_0^2f_s \rv^2 + s^3 \lv Z_0Z_1  f_s \rv^2 \right)\nonumber\\
&+P_{t-s}\left(  
 ((3b+\delta-\frac12  c)s^2 - \lambda (4b-   \frac12\delta )s^3)|Z_1  f_s|^2
\right)\nonumber\\
&+P_{t-s}\left( (a+4\frac{a }{c} +\delta^{-1}c^2
+(\frac{2c^2}{ b}  
+ \frac12 \lambda \delta^{-1}c^2)s  - 2d) |Z_0f_s|^2 \right), \label{BS_1smooth}
\end{align}
with any small $\delta\in(0,\infty)$. From this we see that for sufficiently small $s$ the right hand side is non-positive provided $a$ and $d$ are chosen sufficiently large,
and $6b <  \varepsilon c$ and $c^2 \leq 2ab$. In this situation, given $\sigma\in(0,1)$ we can choose $a$ large enough so that  the following short time smoothing estimate holds
\begin{align*}
 \sigma a t|Z_0f_t|^2+\sigma bt^3|Z_1 f_t|^2 &\leq at|Z_0f_t|^2+bt^3|Z_1  f_t|^2+ ct^2Z_0f_t\cdot Z_1 f_t   
\\ &\leq 
d \left(  P_tf^2 - (P_tf)^2\right)
\end{align*}
for $0<t<t_0$ with some $t_0\in(0,1)$ sufficiently small.
Similar computations yield
\begin{align*}
\partial_sP_{t-s}\big( a |Z_0f_s|^2+b |Z_1  f_s|^2+ &c Z_0f_s\cdot Z_1  f_s +d f_s^2 \big)\\
\leq&  P_{t-s}\left(-2a |Z_0^2 f_s|^2 -\frac12 b |Z_0Z_1  f_s|^2  
-2c Z_0^2f_s\cdot Z_0Z_1  f_s \right)\\
&+P_{t-s}\left(  
 ((3b+\delta-\frac12  c)  - \lambda (4b-   \frac12\delta ) )|Z_1  f_s|^2
\right)\\
&+P_{t-s}\left( (a+4\frac{a}{c} +\delta^{-1}c^2
+(\frac{2c^2}{ b}  
+ \frac12 \lambda \delta^{-1}c^2)   -2d) |Z_0f_s|^2 \right).
\end{align*}
Thus, for sufficiently large $d,a\in(0,\infty)$,  $6b <  \varepsilon c$ and $c^2 \leq 2ab$, the right hand side of the above is non-positive and hence  for any $t\in(0,\infty)$
we have
\begin{align}\label{BS_2smooth}
\sigma a |Z_0f_t|^2+\sigma b |Z_1  f_t|^2  &\leq a |Z_0f_t|^2+b |Z_1  f_t|^2+ c Z_0f_t\cdot Z_1  f_t \\
&   \leq
P_t\left( a |Z_0f |^2+b |Z_1  f |^2+ c Z_0f \cdot Z_1  f\right)  + d \left(P_t f ^2 -  (P_t f )^2 \right) \nonumber\\
&
\leq 2 P_t\left( a |Z_0f |^2+b |Z_1  f |^2\right) + d \left(P_t f ^2 -  (P_t f )^2 \right) \nonumber.
\end{align}
Combining this estimate with the short time smoothing estimate, we arrive at
\begin{eqnarray*}
\sigma a |Z_0f_t|^2+\sigma b |Z_1  f_t|^2   
&\leq 2 P_{t-t_0/2}\left( a |Z_0f_{t_0/2} |^2+b |Z_1  f_{t_0/2} |^2\right) + d \left(P_{t-t_0/2} f_{t_0/2} ^2 -  (P_{t-t_0/2} f_{t_0/2} )^2 \right) \\
&\leq \sigma^{-1} 2  d P_{t-t_0/2}\left(P_{ t_0/2} f^2 -  (P_{t_0/2} f)^2 \right)  + d \left(P_{t-t_0/2} f_{t_0/2} ^2 -  (P_{t-t_0/2} f_{t_0/2} )^2 \right) .
\end{eqnarray*}
Similarly, one can obtain higher order estimates.
The smoothing estimate together with the gradient bounds,
for $\lambda \in(1,\infty)$ imply decay to equilibrium in supremum norm for all continuous functions.
\label{EndSmoothing4BS_dimfin}
We could treat an infinite dimensional version with interaction as follows.
We set
\[
\mathcal{L}\equiv \sum_{j\in\ZZ^d} L_i + \sum_{j,i\in\ZZ^d} G_{ij}x_j \partial_i +\sum_{i\in\ZZ^d} q_i \partial_i
\]
with $L_i$ denoting an isomorphic copy of $L$ acting on $x_i$ and $\partial_i\equiv \partial_{x_i}$, some constants $G_{ij}\in \mathbb{R}$  and some differentiable real functions  $q_i$ with bounded derivatives.
 We note that
 \[\mathcal{L} \langle x_k\rangle \leq \left(\varepsilon + \sup_k \|q_k\|_\infty\right) -\lambda \langle x_k\rangle +   \sum_{j \in\ZZ^d} |G_{kj}| \langle x_j\rangle    \]
and hence for $\epsilon_j\in(0,\infty)$ such that
\[\sum_j\epsilon_j<\infty,\qquad \sum_{j \in\ZZ^d} \epsilon_k |G_{kj}| \leq \xi \epsilon_j\]
with some $\xi\in(0,\infty)$, we obtain
\[\partial_t P_t\sum_k \epsilon_k\langle x_k\rangle \leq 
\left(\varepsilon + \sup_k \|q_k\|_\infty\right)\sum_j\epsilon_j - (\lambda-\xi) \sum_k \epsilon_k\langle x_k\rangle. \]
\label{BS.Cmpct_Infindim}
From this we conclude, that for $\lambda > \xi$
\[ \sup_t \left(P_t \sum_k \epsilon_k\langle x_k\rangle \right) < \infty,  \]
which implies a weak compactness of any sequence
of probability measures $P_{t_n}, n\in\mathbb{N}$, $t_k\to_{k\to\infty}\infty$, on a set
\[ \{\omega\in\mathbb{R}^{\ZZ^d}:\, \sum_k \epsilon_k\langle \omega_k\rangle< \infty \}\]
(see e.g. \cite{daprato,DragoniKontisZegarlinski2011}). We have  
\begin{align*}
\partial_s e^{ms}P_{t-s}\left( \sum_k |\partial_k f_s|^2 \right)  \leq&  e^{ms} P_{t-s}\left(-2e^{ms}\sum_{j\neq k}\left(Z_j\partial_kf_s\right)^2  \right) \\
&- 2 e^{ms} P_{t-s}\left(
 \sum_{k}  \left(Z_k \partial_kf_s -\partial_kf_s\right)^2\right)   
  \\&-  (2 (\lambda  - 1) - m)  e^{ms} P_{t-s}\left( \sum_{k} |\partial_k f_s|^2 \right) \\
&+2 e^{ms} P_{t-s}\left(
 \sum_{k,i}  \left(G_{ik} + \partial_k q_i \right)\partial_k f_s \partial_i f_s
\right), 
\end{align*}
and hence
\[
 \partial_s e^{ms}P_{t-s}\left( \sum_k |\partial_k f_s|^2 \right)  \leq   
   - C  e^{ms} P_{t-s}\left( \sum_{k} |\partial_k f_s|^2 \right)  
\]
with
\[C \equiv 2 (\lambda  - 1)  -m - \sup_i\sum_{k}  \left(|G_{ik}| +|G_{ki}|+ \|\partial_k q_i\|_\infty+\|\partial_i q_k\|_\infty \right) 
.\]
Hence, if
\[ 0 \leq m\leq 2 (\lambda  - 1)   - \sup_i\sum_{k}  \left(|G_{ik}| +|G_{ki}|+ \|\partial_k q_i\|_\infty+\|\partial_i q_k\|_\infty \right) ,\]
we get
\[\sum_k |\partial_k f_t|^2  \leq e^{-mt} P_t \sum_k |\partial_k f|^2 .  \]

We consider only smoothing estimate in infinite dimensions for the case when $G_{ij}\equiv 0$. In this case  we have  
\label{BeginSmoothing4BS_Dim_Infin} 
\begin{align*}
\partial_s&P_{t-s}\left( \sum_k\left( as|Z_{0,k}f_s|^2+bs^3|Z_{1,k} f_s|^2+ cs^2Z_{0,k}f_s\cdot Z_{1,k} f_s \right) +d f_s^2 \right)\\
=&P_{t-s}\sum_{j}\left( \sum_k\left(-2as|Z_{0,j}Z_{0,k}f_s|^2 -2bs^3|Z_{0,j} Z_{1,k} f_s|^2  
-2cs^2Z_{0,j}Z_{0,k} f_s\cdot Z_{0,j}Z_{1,k} f_s\right) -2d |Z_{0,j}f_s|^2 \right)\\
&+P_{t-s}\sum_k\left(  - 2as Z_{0,k} f_s\cdot  Z_{1,k}f_s 
+2 bs^3 Z_{1,k}f_s\cdot (\{Z_{0,k},Z_{1,k}\} -(\lambda+1)Z_{1,k})f_s \right)\\
&+P_{t-s}\sum_k\left( -  cs^2 |Z_{1,k} f_s|^2 +  cs^2 Z_{0,k}f_s\cdot  (\{Z_{0,k}, Z_{1,k}\} -(\lambda+1)Z_{1,k})f_s )
\right)\\
&+P_{t-s}\sum_k\left( a |Z_{0,k}f_s|^2+3bs^2|Z_{1,k} f_s|^2+ 2cs Z_{0,k}f_s\cdot Z_{1,k}  f_s     \right)\\
&
+\sum_{i\in\ZZ^d} P_{t-s}\sum_k\left( 2as Z_{0,k} f_s\cdot [ Z_{0,k} ,  q_i \partial_i]f_s \right)\\
&
+ \sum_{i\in\ZZ^d} P_{t-s}\sum_k\left( 2b s^3 Z_{1,k} f_s\cdot [Z_{1,k} ,q_i \partial_i ]f_s\right)\\
&
+\sum_{i\in\ZZ^d} P_{t-s}\sum_k\left(  cs^2  [ Z_{0,k} , q_i \partial_i]f_s\cdot Z_{1,k}f_s\right) \\
&
+\sum_{i\in\ZZ^d} P_{t-s}\sum_k\left(  cs^2   Z_{0,k} f_s\cdot [Z_{1,k}, q_i \partial_i]f_s\right). 
\end{align*}
Next, we estimate the new type of terms involving $q_i$'s, i.e. the last four terms above. We have\begin{align*}
\sum&_{i\in\ZZ^d} P_{t-s}\sum_k\left( 2as Z_{0,k} f_s\cdot [ Z_{0,k} ,  q_i \partial_i]f_s \right)
\\ &=\sum_{i\in\ZZ^d} P_{t-s}\sum_k \left( 2as Z_{0,k} f_s\cdot
\left( \varepsilon^{-1}(Z_{0,k} q_i ) Z_{1,i} f_s - \delta_{ik} q_k Z_{0,k} f_s\right) \right)\\
 &\leq \left( \delta^{-1}\varepsilon^{-1}
 \sup_j\sum_{i\in\ZZ^d}\|Z_{0,j} q_i \|_\infty\right) a^2  P_{t-s}\sum_k  |Z_{0,k} f_s|^2 
 +\left(\delta \sup_j\sum_i\|Z_{0,i} q_j \|_\infty\right) s^3 P_{t-s}\sum_{k\in\ZZ^d} |Z_{1,k} f_s|^2 \\
 &\quad+ \left( 2a \sup_i\|q_i\|_\infty\right) s  
 P_{t-s}\sum_k    |Z_{0,k} f_s|^2, 
 \end{align*}
 while
the second term can be estimated as \begin{align*}
 \sum&_{i\in\ZZ^d} P_{t-s}\sum_k\left( 2b s^3 Z_{1,k} f_s\cdot [Z_{1,k} ,q_i \partial_i ]f_s\right)\\& =
 \sum_{i\in\ZZ^d} P_{t-s}\sum_k\left( 2b s^3 Z_{1,k} f_s
\cdot \left(\varepsilon^{-1} (Z_{1,k}  q_i) Z_{1,i} \right)f_s  \right)
\\
& \leq  \varepsilon^{-1}b s^3\left( \sup_i \sum_j  \|Z_{1,j}  q_i\|_\infty +\sup_i \sum_j  \|Z_{1,i}  q_j\|_\infty \right) P_{t-s} \sum_k |Z_{1,k} f_s|^2 .
 \end{align*}
 Moreover, we have
 \begin{align*}
 \sum&_{i\in\ZZ^d} P_{t-s}\sum_k\left(  cs^2  [ Z_{0,k} , q_i \partial_i]f_s\cdot Z_{1,k}f_s\right)\\&=
\sum_{i\in\ZZ^d} P_{t-s}\sum_k\left(  \varepsilon^{-1} cs^2 \left( (Z_{0,k} q_i ) Z_{1,i} f_s - \delta_{ik} q_k Z_{1,k} f_s\right)\cdot Z_{1,k}f_s\right)   \\
&\leq \varepsilon^{-1} cs^2 \left( \sup_k  \|q_k\|_\infty + \sup_j\sum_{i\in\ZZ^d}  \varepsilon^{-1}\|Z_{0,j} q_i\|_\infty + \sup_i\sum_{j\in\ZZ^d}  \varepsilon^{-1}\|Z_{0,j} q_i\|_\infty \right) 
P_{t-s}\sum_k  |Z_{1,k} f_s|^2 
 \end{align*}
 and finally
 \begin{align*}
\sum&_{i\in\ZZ^d} P_{t-s}\sum_k\left(  cs^2   Z_{0,k} f_s\cdot [Z_{1,k}, q_i \partial_i]f_s\right) \\& =
\sum_{i\in\ZZ^d} P_{t-s}\sum_k\left(  cs^2   Z_{0,k} f_s\cdot  \left(\varepsilon^{-1} (Z_{1,k}  q_i) Z_{1,i} \right)f_s \right)  \\
&\leq 
\delta^{-1} \varepsilon^{-2}c^2s   \sup_j\sum_{i\in\ZZ^d}\|Z_{1,j}  q_i\|_\infty   P_{t-s}\sum_k |Z_{0,k} f_s|^2 \\
&\quad + \delta  s^3   \sup_i\sum_{j\in\ZZ^d}\|Z_{1,j}  q_i\|_\infty 
\cdot  P_{t-s}\sum_k |Z_{1,k} f_s|^2 ,
\end{align*}
for any $\delta\in(0,\infty)$.
Combining this together with \eqref{BS_1smooth} and assuming 
 \begin{eqnarray*}
&\sup_i \sum_j  \|Z_{1,j}  q_i\|_\infty +\sup_i \sum_j  \|Z_{1,i}  q_j\|_\infty  < \infty,\\
& \sup_k  \|q_k\|_\infty < \infty,\\
& \sup_j\sum_{i\in\ZZ^d}   \|Z_{0,j} q_i\|_\infty + \sup_i\sum_{j\in\ZZ^d}   \|Z_{0,j} q_i\|_\infty < \infty
\end{eqnarray*}
are sufficiently small, we conclude the short time smoothing.
Repeating this computation without algebraic coefficients in $s$, one obtains global gradient bounds which together 
with small time smoothing provides global smoothing estimates.
\label{EndSmooth4BS_Dim_Infin} \\

\begin{remark}[Grushin-type Operators] \label{Grushin Type Ops}  
Similarly, one can handle infinite dimensional models build upon a   model in $\mathbb{R}^k\times\mathbb{R}^n$ with generators of the following type
\[L= \Delta_x + |x|^2\Delta_y 
+\varepsilon 1\cdot\nabla_x-\lambda x\cdot\nabla_x.
\]
If the coefficients at the principal part are of higher order 
as in the following cases: 

(i) In $\mathbb{R}$
\[L= x^{2m}\partial_x^2 +\partial_x -\lambda D
\]
with $1<m\in\mathbb{N}$;\\

(ii) In $\mathbb{R}^k\times\mathbb{R}^n$
\[L= \Delta_x + |x|^{2m}\Delta_y 
+\varepsilon 1\cdot\nabla_x - \lambda x\cdot\nabla_x
\]
with $1<m\in\mathbb{N}$,
then the corresponding algebra becomes infinite and 
more involved arguments are necessary.
 \end{remark} 

\end{document}